\documentclass[aps,showpacs,amssymb,amsfonts,superscriptaddress,twocolumn,prl]{revtex4-1}
\usepackage{bm,bbm}
\usepackage{times}
\usepackage{graphicx,amsthm,epstopdf,amsmath}
\usepackage{subfigure,color}
\usepackage[draft]{fixme}
\usepackage{changes}

\definecolor{myurlcolor}{rgb}{0,0,0.7}
\definecolor{myrefcolor}{rgb}{0.8,0,0}
\usepackage{hyperref}
\hypersetup{colorlinks, linkcolor=myrefcolor,
citecolor=myurlcolor, urlcolor=myurlcolor}

\newtheorem{thm}{Theorem}
\newtheorem{thmm}{Theorem}

\newtheorem{faktt}[thmm]{Fact}

\newcommand{\beu}{\begin{equation}}
\newcommand{\eeu}{\end{equation}}

\newcommand{\be}{\begin{eqnarray}}
\newcommand{\ee}{\end{eqnarray}}

\newcommand{\ba}{\begin{array}}
\newcommand{\ea}{\end{array}}

\newcommand{\ket}[1]{|#1\rangle}

\newcommand{\Tr}[0]{\mathrm{Tr}}
\newcommand{\sr}[1]{\langle[#1]\rangle}
\newcommand{\nms}{\negmedspace}

\begin{document}

\title{Elemental and tight monogamy relations in nonsignalling theories}


\author{R. Augusiak}
\affiliation{ICFO--Institut de Ci\`{e}ncies Fot\`{o}niques, 08860
Castelldefels (Barcelona), Spain}

\author{M. Demianowicz}
\affiliation{ICFO--Institut de Ci\`{e}ncies Fot\`{o}niques, 08860
Castelldefels (Barcelona), Spain}

\author{M. Paw\l{}owski}
\affiliation{Instytut Fizyki Teoretycznej i
Astrofizyki, Uniwersytet Gda\'nski, PL-80-952 Gda\'nsk, Poland}

\author{J. Tura}
\affiliation{ICFO--Institut de Ci\`{e}ncies Fot\`{o}niques, 08860
Castelldefels (Barcelona), Spain}

\author{A. Ac\'in}
\affiliation{ICFO--Institut de Ci\`{e}ncies Fot\`{o}niques, 08860
Castelldefels (Barcelona), Spain}
\affiliation{ICREA--Instituci\'o Catalana de Recerca
i Estudis Avan\c{c}ats, Lluis Companys 23, 08010 Barcelona, Spain}

\begin{abstract}
Physical principles constrain the way nonlocal correlations can be
distributed among distant parties. These constraints are usually
expressed by monogamy relations that bound the amount of Bell
inequality violation observed among a set of parties by the
violation observed by a different set of parties. We prove here
that much stronger monogamy relations are possible for
nonsignalling correlations by showing how nonlocal correlations
among a set of parties limit \textit{any} form of correlations,
not necessarily nonlocal, shared among other parties. In
particular, we provide tight bounds between the violation of a
family of Bell inequalities among an arbitrary number of parties
and the knowledge an external observer can gain about outcomes of
\textit{any single} measurement performed by the parties. Finally,
we show how the obtained monogamy relations offer an improvement
over the existing protocols for device-independent quantum
key distribution and randomness amplification.
\end{abstract}
\maketitle

\textit{Introduction.} It is a well established fact that
entanglement and nonlocal correlations (cf. Refs.
\cite{Hreview,NLreview}), i.e., correlations violating a Bell
inequality \cite{Bell}, are fundamental resources of quantum
information theory. It has been confirmed by many instances that,
when distributed among spatially separated observers, they give an
advantage over classical correlations at certain
information-theoretic tasks, many of them being considered in the
multipartite scenario. For instance, nonlocal correlations
outperform their classical counterpart at communication complexity problems
\cite{CommCompl}, and allow for security not achievable within
classical theory \cite{crypto}.

Physical principles impose certain constraints on the way these
resources can be distributed among separated parties; these are
commonly referred to as monogamy relations. For instance, in any
three--qubit pure state one party cannot share large amount of
entanglement, as measured by concurrence, simultaneously with both
remaining parties \cite{Wootters}. Analogous monogamy relations,
both in qualitative \cite{NS,BarrettMon,BKP,rodrigo} and
quantitative \cite{TV,T} form, were demonstrated for nonlocal
correlations, with the measure of nonlocality being the violation
of specific Bell inequalities. In particular, Toner and Verstraete
\cite{TV} and later Toner \cite{T} showed that if three parties
$A$, $B$, and $C$ share, respectively, quantum and general
nonsignalling correlations, then only a single pair can violate
the Clauser-Horne-Shimony-Holt (CHSH) Bell inequality \cite{CHSH}.
These findings were generalized to more complex scenarios
\cite{ns-mon,Pawel} (see also Ref. \cite{Thiago}), and in
particular in \cite{ns-mon} a general construction of monogamy
relations for nonsignalling correlations from any bipartite Bell
inequality was proposed.

In this work, we demonstrate that nonsignalling correlations are
monogamous in a much stronger sense: the amount of nonlocality
observed by a set of parties may imply severe limitations on any
form of correlations with other parties. That is, instead of
comparing nonlocality between distinct groups of parties, we
rather relate it to the knowledge that external
parties can gain on outcomes of any of the measurements performed
by the parties (see Fig. \ref{fig:concept}). To be more
illustrative, consider again parties $A$, $B$, and $C$ performing
a Bell experiment with $M$ observables and $d$ outcomes. We
construct tight bounds between the violation of certain Bell
inequalities \cite{BKP} among any pair of parties, say $A$ and
$B$, and classical correlations that the third party $C$ can
establish with outcomes of any measurement performed by $A$ or
$B$. This means that the amount of \textit{any} correlations ---
classical or nonlocal --- that $C$ could share with $A$ or $B$ is
bounded by the strength of the Bell inequality violation between
$A$ and $B$. Our monogamies are further generalized to the
scenario with an arbitrary number of parties $N$ [$(N,M,d)$
scenario] with nonlocality measured by the recent generalization
of the Bell inequalities \cite{BKP} presented in Ref.
\cite{rodrigo}. The obtained monogamy relations are logically
independent from, and in fact stronger than, the existing
relations involving only nonlocal correlations, as a bound on
nonlocal correlations does not necessarily imply any nontrivial
constraint on the amount of classical correlations.

\begin{figure}[t]
(a)\includegraphics[width=0.175\textwidth]{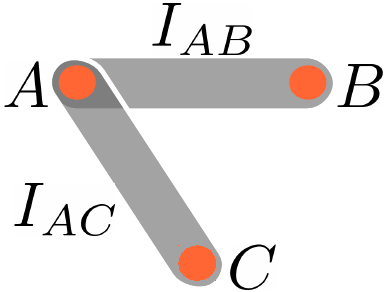}
(b)\includegraphics[width=0.2\textwidth]{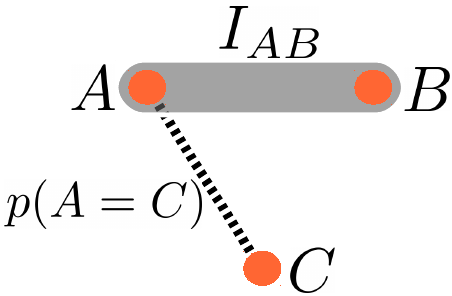}
\caption{
%
(a) The usual monogamies compare nonlocality (measured by the
value of some Bell inequality $I$) between different groups of
parties (here between two pairs of parties $AB$ and $AC$).
Instead, our monogamy relations compare nonlocality observed by a
group of parties (here $AB$) to the knowledge, represented by the
probability $p(A=C)$, that the third party $C$ can have about
outcomes observed by either of the parties. As such, they are
qualitatively different, and in fact stronger than those of type
(a).}\label{fig:concept}
\end{figure}


Our new monogamy relations prove useful in device-independent
protocols \cite{device-ind,DIRNG,collbeck-renner,andrzej,qit}. First, we
show that they impose tight bounds on the guessing probability,
the commonly used measure of randomness, that are significantly
better than the existing ones \cite{BKP,rodrigo}. We then argue
that this translates into superior performance in protocols for
device-independent quantum key distribution (DIQKD) \cite{Lluis2}
using measurements with more than two outputs. Finally, we show
that they allow for a generalization of the results of
\cite{collbeck-renner} on randomness amplification to any number
of parties and outcomes, demonstrating, in particular, that
arbitrary amount of arbitrarily good randomness can be amplified
in a bipartite setup.

Before turning to the results, we provide some background.
Consider $N$ parties $A^{(1)},\ldots,A^{(N)}$ (for $N=3$
denoted by $A,B,C$), each measuring one of $M$
possible observables $A^{(i)}_{x_i}$ $(x_i=1,\ldots,M)$ with $d$
outcomes (enumerated by $a_i=1,\ldots,d$) on their local physical
systems. The produced correlations are described by a
collection of probabilities $p(A^{(1)}_{x_1}=a_1,\ldots,
A^{(N)}_{x_{N}}=a_{N})\equiv p(a_1\ldots a_{N}|x_1\ldots x_{N})\equiv
p(\boldsymbol{a}|\boldsymbol{x} )$ of obtaining
results $\boldsymbol{a}\equiv a_1\ldots a_{N}$
upon measuring $\boldsymbol{x}\equiv x_1\ldots x_{N}$.
One then says that the correlations
$\{p(\boldsymbol{a}|\boldsymbol{x})\}$ are (i) nonsignalling
(NC) if any of the marginals describing a subset of parties is
independent of the measurements choices made by the remaining
parties and (ii) quantum (QC) if they arise by local measurements
on quantum states (cf. \cite{NLreview}).


\textit{Elemental and tight monogamies for nonsignalling
correlations.} We start with the derivation of our monogamy
relations in the case of nonsignalling correlations. For
clarity, we begin with the simplest tripartite scenario. We will use the Bell inequality introduced by Barrett, Kent, and Pironio (BKP) \cite{BKP}. Denoting by $\langle \Omega\rangle$ the mean value of a random variable $\Omega$, that is, $\langle
\Omega\rangle=\sum_{i=1}^{d-1}iP(\Omega=i)$, it reads
\begin{equation}\label{BKPNd}
I^{2,M,d}_{AB}:=\sum_{\alpha=1}^{M}\left(\sr{A_{\alpha}-B_{\alpha}}
+\sr{B_{\alpha}-A_{\alpha+1}}\right)\geq
d-1
\end{equation}
with $[\Omega]$ being $\Omega$ modulo $d$, and
$\Omega_{M+1}:=[\Omega_1+1]$. For $d=2$, Ineq.
(\ref{BKPNd}) reproduces the chained Bell inequalities \cite{BC},
while for $M=2$ the Collins-Gisin-Linden-Massar-Popescu (CGLMP)
inequalities \cite{Collins}. The maximal nonsignalling violation
of (\ref{BKPNd}) is $I_{AB}^{2,M,d}=0$.

The only monogamy relations for (\ref{BKPNd}) have been formulated
in terms of its violations between Alice and $M$ Bobs
\cite{ns-mon}, which is a natural quantitative extension of the
concept of $M$-shareability \cite{NS}. In the
following theorem we show that the BKP Bell inequalities allow one
to introduce \textit{elemental} monogamies obeyed
by any NC.

\begin{thm}\label{thm:BKPNd}
For any tripartite NC $\{p(abc|xyz)\}$ with $M$ $d$-outcome
measurements, the inequality
\begin{equation}\label{thm:BKPNd:1}
I^{2,M,d}_{AB}+\sr{X_i-C_j}+\sr{C_j-X_i}\geq d-1
\end{equation}
holds for any pair
$i,j=1,\ldots,M$ and $X$ denoting $A$ or $B$.
\end{thm}

Interestingly, all these inequalities are tight in the sense that
for any values of $I_{AB}^{2,M,d}$ and $\sr{X_i-C_j}+\sr{C_j-X_i}$
saturating (\ref{thm:BKPNd:1}), one can find
NC realizing these values. Take, for instance, a probability distribution
$\{p(a,b,c|x,y,z)=p(a,b|x,y)p(c|z)\}$, with
$\{p(a,b|x,y)\}$ being a mixture of a nonlocal
model maximally violating (\ref{BKPNd})
%
%
and a local deterministic one saturating it. Then, $\{p(c|z)\}$
is the same distribution as the one used by $A$ or $B$ in the local model
saturating (\ref{BKPNd}).
%

The physical interpretation of our monogamies can be now
concluded if we rewrite them in a bit different form. Using the
fact that for any variable $\Omega$,
$\sr{\Omega}+\sr{-\Omega}=dP([\Omega]\neq 0)=d[1-P([\Omega]=0)]$
\cite{supplement}, Ineqs. (\ref{thm:BKPNd:1}) transform to
\begin{equation}\label{monogII}
I^{2,M,d}_{AB}+1\geq
dp(X_{i}=C_{j})
\end{equation}
for $X=A,B$, and any pair $i,j=1,\ldots,M$. These relations hold
if $AB$ is replaced by any pair of parties and if any
$m=1,\ldots,d-1$ is added modulo $d$ to the argument of
probability. The meaning of the introduced monogamy relations is
now transparent. The probability $p(X_{i}=C_{j})$ that parties $X$
and $C$ obtain the same results upon measuring the $i$th and $j$th
observables is a measure of how the outcomes of these measurements
are classically correlated. Consequently, Ineqs.
(\ref{thm:BKPNd:1}) establish trade-offs between nonlocality, as
measured by (\ref{BKPNd}), that can
be generated between any two parties and classical correlations
that the third party can share with the results of any measurement
performed by any of these two parties. Furthermore, they are
tight. In fact, it is known that the maximal NC violation of
(\ref{BKPNd}), $I^{2,M,d}_{AB}=0$, implies $p(X_i=C_{j})=1/d$ for
any $i,j=1,\ldots,M$, meaning that at the point of maximal
violation $C$ cannot share any correlations with any other party's
measurement outcomes~\cite{BKP}. On the other hand, it is well
known that at the point of no violation $C$ can be arbitrarily
correlated with $A$ and $B$. For intermediate violations, the best
one can hope for is a linear interpolation between these two
extreme values and this is precisely what our monogamy relations
predict, see Fig.~\ref{fig:comparison}.

Let us now move to the general case of an arbitrary number of
parties each having $M$ $d$-outcome observables at their disposal.
We will utilize the generalization of the Bell inequality
(\ref{BKPNd}) introduced in Ref. \cite{rodrigo}, which can be
stated as
\begin{equation}\label{BKPNMd}
I_{\mathsf{A}}^{N,M,d}\geq d-1
\end{equation}
%
with $\mathsf{A}=A^{(1)}\ldots A^{(N)}$. Since the form of
$I_{\mathsf{A}}^{N,M,d}$ is rather lengthy and actually not
relevant for further considerations, for
clarity, we omit presenting it here (see
\cite{supplement}). We only mention that it can be recursively determined from $I_{AB}^{2,M,d}$ and that its minimal nonsignalling value is
$I_{\mathsf{A}}^{N,M,d}=0$. Then, the generalization of Theorem
\ref{thm:BKPNd} to arbitrary $N$ goes as follows.
\begin{thm}\label{thm:BKPNMd}
For any $(N+1)$--partite NC
$\{p(\boldsymbol{a}|\boldsymbol{x})\}$ with $M$ $d$-outcome measurements
per site, the following inequality
\begin{equation}\label{thm:BKPNMd:1}
I^{N,M,d}_{\mathsf{A}}+\sr{A_{x_k}^{(k)}-A^{(N+1)}_{x_{N+1}}}+\sr{A^{(N+1)}_{x_{N+1}}-A_{x_k}^{(k)
} } \geq d-1
\end{equation}
is satisfied for any $x_k,x_{N+1}=1,\ldots,M$ and $k=1,\ldots,N$.
\end{thm}

All the properties of the three-partite monogamy relations
persist for any $N$. In particular, all
inequalities (\ref{thm:BKPNMd:1}) are tight. Moreover, they can be
rewritten as
\begin{equation}\label{monogII}
I^{N,M,d}_{\mathsf{A}}+1\geq
dp(A^{(k)}_{x_k}=[A^{(N+1)}_{x_{N+1}}+m])
\end{equation}
for any $x_k,x_{N+1}=1,\ldots,M$, $k=1,\ldots,N$ and
$m=0,\ldots,d-1$ and remain valid if the nonlocality is tested
among any $N$-element subset of $N+1$ parties. Analogously to the
three-partite case, Ineqs. (\ref{monogII}) tightly relate the
nonlocality observed by any $N$ parties, as measured by
$I_{\mathsf{A}}^{N,M,d}$, and correlations that party $(N+1)$ can
share between measurement outcomes of any of these $N$ parties. It
is worth pointing out that for $d=2$ it holds
$\sr{X-Y}=\sr{Y-X}$, and Ineqs. (\ref{thm:BKPNMd:1}) simplify to
$I^{N,M,2}_{\mathsf{A}}+2\sr{A_{x_k}^{(k)}-A^{(N+1)}_{x_{N+1}}}\geq
1$ which can be rewritten in a more familiar form as
%
%
$|\langle A_{x_k}^{(k)}A^{(N+1)}_{x_{N+1}}\rangle|\leq
I^{N,M,2}_{\mathsf{A}}$,
%
where $A_{x_k}^{(k)}$ stand now
for dichotomic observables with outcomes $\pm1$, while $\langle
XY\rangle=P(X=Y)-P(X\neq Y)$. Thus, the strength of violation of
%
%
(\ref{BKPNMd}) imposes tight bounds on a \textit{single} mean
value $\langle A_{x_k}^{(k)}A^{(N+1)}_{x_{N+1}}\rangle$ for any
$x_k,x_{N+1}$ and $k=1,\ldots,N$, which is also a measure of how
outcomes of a measurement performed by the external party
$A^{(N+1)}$ are correlated to those of $A^{(k)}$ for any $k$. In
particular, when $I_{\mathsf{A}}^{N,M,2}=0$ (maximal
nonsignalling violation), all these means are zero,
while maximal correlations between a single pair of measurements,
i.e., $\langle A^{(k)}_{x_{k}}A^{(N+1)}_{x_{N+1}}\rangle=\pm1$ for
some $x_k,x_{N+1}$, make the $N$ parties unable to violate
$I_{\mathsf{A}}^{N,M,2}\ge 1$.

\textit{Bounds on randomness.} Our monogamies are of particular
importance for device--independent applications since they
imply upper bounds on the guessing probability (GP) of the
outcomes of any measurement performed by any of the $N$ parties by
the additional party, here called $E$. To be precise, assume that
$E$ has full knowledge about all parties devices and their
measurement choices and wishes to guess the outcomes of, say
$A_{x_k}^{(k)}$. The best $E$ can do for this purpose is to simply
measure one of its observables, say the $z$th one, and,
irrespectively of the obtained result, deliver the most probable
outcome of $A^{(k)}_{x_k}$. Then,
$\max_{a_k}p(A_{x_k}^{(k)}=a_k)=p(E_{z}=A^{(k)}_{x_k})$, and
Ineqs. (\ref{monogII}) imply that for any $x_k$ and $k$, GP is
bounded as
\begin{equation}\label{guessing}
\max_{a_k}p(a_k|x_k)\equiv\max_{a_k}p(A_{x_k}^{(k)}=a_k)\leq
\frac{1}{d}(1+I_{\mathsf{A}}^{N,M,d}).
\end{equation}
These bounds are tight and significantly stronger
than the previously existing one,
\begin{equation}\label{guessing-bkp} \max_{a_k}p(a_k|x_k)\leq
\frac{1}{d}\left(1+\frac{d^N}{4}(N-1)I_{\mathsf{A}}^{N,M,d}\right)
\end{equation}
derived in Refs. \cite{BKP,rodrigo} (see Fig.
\ref{fig:comparison}).
\begin{figure}[t]
(a)\includegraphics[width=0.225\textwidth]{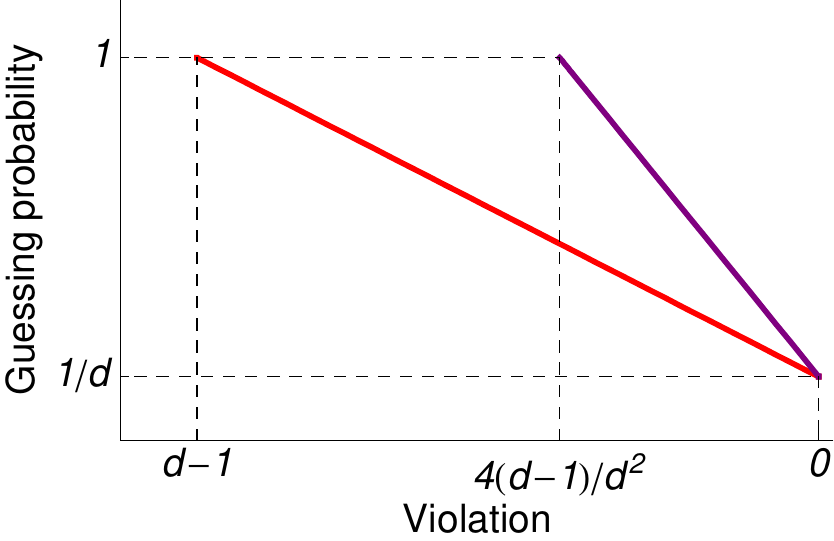}
(b)\includegraphics[width=0.205\textwidth,trim=0 -18.7 0
0]{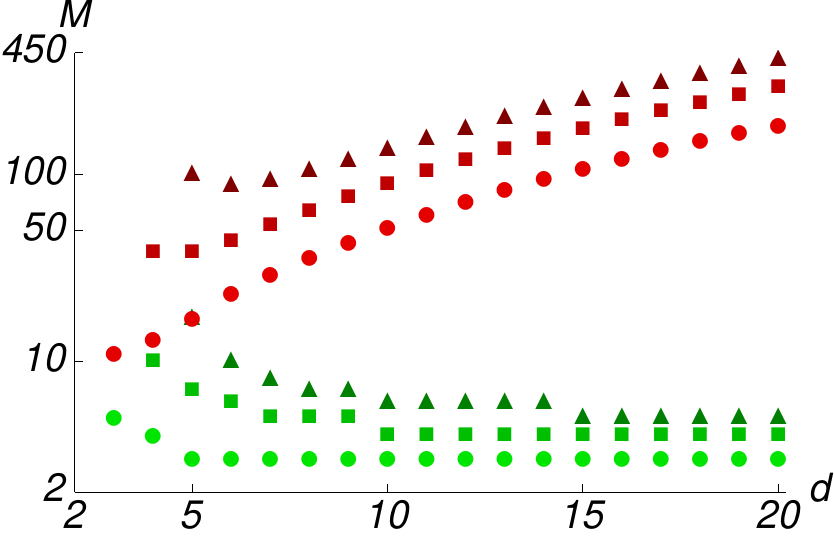} \caption{(a) Comparison of the upper bounds on
GP: present bound (\ref{guessing}) (red line) and
(\ref{guessing-bkp}) (purple line).
%
%
Our bound  is tight -- for any value $0\leq
I_{\mathsf{A}}^{N,M,d}\leq d-1$, it provides the maximum
attainable value of GP. Instead, the bound (\ref{guessing-bkp}) is
nontrivial only in some restricted range of
$I_{\mathsf{A}}^{N,M,d}$, namely when $I_{\mathsf{A}}^{N,M,d} <
4(d-1)/d^2$, which tends to zero for $d\to\infty$. (b) Minimal
number of measurements $M$ on a maximally entangled state of local
dimension $d$ necessary for the secret-key rate $R$ secure against
non-signalling eavesdroppers to be at least:
one (dots), $\log_23$ (squares), and two (triangles) bits, when
(\ref{guessing}) (green) and (\ref{guessing-bkp}) (red) are used
to bound $R$. Using our bound the parties need to use many fewer
measurements to reach the same key rate. Moreover, contrary to
what is predicted by the previous bound, the number of measurement
decreases with the dimension.
}\label{fig:comparison}
\end{figure}

Let us now discuss how the bound (\ref{guessing}) performs in
comparison to (\ref{guessing-bkp}) in security proofs of DIQKD
against no-signalling eavesdroppers. At the moment, a general
security proof in this scenario is missing and the strongest proof
requires the assumption that the eavesdropper $E$ is not only
limited by the no-signalling principle but also lacks a long-term
quantum memory (so--called bounded-storage model) \cite{Lluis2}. 
Assume that Alice and Bob share a two-qudit maximally entangled state and they
use it to maximally violate (\ref{BKPNd}) by performing the optimal
measurements for this setup (see, e.g., \cite{BKP}). To generate
the secure key, Bob performs one more measurement that is
perfectly correlated to one of Alice's measurements. The key
rate of this protocol is lower-bounded as $R\geq
-\log_2[\tau(I_{AB}^{2,M,d})]-H(A|B)$ \cite{Lluis2}, where $\tau$ is any
upper bound on GP for nonsignalling correlations.
and $H(A|B)$ is the conditional Shannon entropy between
Alice and Bob for the measurements used to generate the secret
key. As the state is maximally entangled, this term is equal to zero.
Fig. \ref{fig:comparison} compares bounds on the secret key
obtained by using our bound (\ref{guessing}) and the previous
bound (\ref{guessing-bkp}) in this protocol. We fix the key rate
and compute the minimal number of measurements needed to attain
this rate using these bounds as a function of the number of
outputs. As shown in Fig. \ref{fig:comparison}, the number of
measurements when using our bound is much smaller and, in
particular, decreases with the number of outputs.

\textit{Randomness amplification.} Let us finally show the usefulness
of our monogamy relations in randomness amplification.
Assume that each party is given a sequence of bits produced by the Santha--Vazirani (SV) source (or the $\varepsilon$--source). Its working is defined as follows: it produces a sequence $y_1, y_2,\cdots, y_n$ of bits according to
\begin{equation}\label{sv}
\tfrac{1}{2}-\varepsilon \le p(y_k| w)\le
\tfrac{1}{2}+\varepsilon, \quad k=1,\ldots,n,
\end{equation}
where $w$ denotes any space-time variable that could be the cause
of $y_k$. Thus the bits are possibly correlated with each other
retaining, however, some intrinsic randomness --- we say that they
are $\varepsilon$--free. The goal is now to obtain a perfectly
random bit (or more generally $d$it) from an arbitrarily long
sequence of $\varepsilon$--free bits by using quantum correlations
that violate the Bell inequality (\ref{BKPNMd}).
This procedure is called randomness amplification (RA).

It is useful to recast this task in the adversarial picture
\cite{collbeck-renner}, in which one assumes that an adversary
$E$, using the $\varepsilon$--sources, wants to simulate the
quantum violation of (\ref{BKPNMd}) by NC, in
particular the local ones. The random variable $W$ is now held by
$E$ who uses it to control both the $\varepsilon$-sources and the
physical devices possessed by the parties. That is, for every
value $w$ of $W$ the former provides settings $\boldsymbol{x}$
with probabilities obeying (\ref{sv}), while these devices
generate the $N$-partite probability distribution
$\{p(\boldsymbol{a}|\boldsymbol{x},w)\}_{\boldsymbol{a},\boldsymbol{x}}$.
Using (\ref{guessing}), we can now restate and generalize Lemma 1
of \cite{collbeck-renner} (see \cite{supplement}).

\begin{thm}Let
$\{p(\boldsymbol{a}|\boldsymbol{x},w)\}_{\boldsymbol{a},\boldsymbol{x}}$ be
a nonsignalling probability distribution for any $w$. Then
for any $\boldsymbol{x}$ and $k=1,\ldots,N$:
\begin{equation}\label{VIB}
\sum_{a_k,w}\left|p(a_k,w|\boldsymbol{x})-\widetilde{p}(a_k)p(w|\boldsymbol{x})\right|
\leq
\tfrac{(d-1)^2+1}{d}\,Q_M(\boldsymbol{x})
I^{N,M,d}_{\mathsf{A}},
\end{equation}
where $\widetilde{p}(a)=1/d$ for any $a$,
$\{p(a_k,w|\boldsymbol{x})\}_{a_k,w}$ describes correlations between outcomes obtained
by party $k$ and the random variable $W$ for the measurements choice $\boldsymbol{x}$, and
$I^{N,M,d}_{\mathsf{A}}$ is taken in the probability distribution
$\{p(\boldsymbol{a}|\boldsymbol{x})\}$ observed by the parties. Finally,
$Q_M(\boldsymbol{x})=\max_{w}[p(w|\boldsymbol{x})/p_{\min}(w) ] $,
where $p_{\min}(w)=\min_{\boldsymbol{x}}\{p(w|\boldsymbol{x})\}$ with
minimum taken over those measurement settings $\boldsymbol{x}$
that appear in $I_{\mathsf{A}}^{N,M,d}$.
%
\end{thm}

It then follows that if correlations $\{p(\boldsymbol{a}|\boldsymbol{x})\}$
violate maximally the Bell inequality (\ref{BKPNMd}),
then the \textit{dits} observed by the parties are perfectly random and
uncorrelated from $W$ \cite{collbeck-renner}.

Let us now show that one can amplify partially random input
bits
%
%
to almost perfectly
random \textit{dits} by using QC that produce
arbitrarily high violation of $I_{\mathsf{A}}^{N,M,d}$.
%
%
To generate one of the $M$ measurement settings, each party uses
its SV source $r=\lceil\log_2M\rceil$ times. Hence for any
$\boldsymbol{x}$, $Q_r(\boldsymbol{x})\leq
[(1+2\varepsilon)/(1-2\varepsilon)]^{Nr}$ (cf. Ref.
\cite{collbeck-renner}). Then, there is a state
and measurement settings \cite{BKP,rodrigo} such that for large
$M$,
\begin{equation}\label{viol}
I^{N,M,d}_{\mathsf{A}}\approx
\lambda(d)/M\leq \lambda(d)/2^{r-1},
\end{equation}
where $\lambda(d)$ is a function of $d$. After
plugging everything into (\ref{VIB}), one checks that its
r.h.s. tends to zero for $M\to\infty$ iff
$\varepsilon<\varepsilon_N:=(2^{1/N}-1)/[2(2^{1/N}+1)]$. As a
result, QC violating (\ref{viol}) can be used to amplify
randomness of any $\varepsilon$-source provided
$\varepsilon<\varepsilon_N$. In particular, for $N=2$, the above
reproduces the value $\varepsilon_2=(\sqrt{2}-1)^2/2$ found in
\cite{collbeck-renner}, and, because $\varepsilon_N$ is a strictly
decreasing function of $N$, the larger $N$, the lower the critical
epsilon $\varepsilon_N$ for this method to work. Notice, however,
that $\varepsilon_N$ is independent of $d$, so almost perfectly
random \textit{dits} are obtained from partially random bits. This
means that using the setup from Ref. \cite{collbeck-renner} we can
in fact achieve both amplification and expansion of randomness
simultaneously.

Recently, with the same Bell inequality but for $N=d=2$, the
critical epsilon was shifted from $\varepsilon_2\approx0.086$ to
$\varepsilon_2'\approx 0.0961$ \cite{andrzej}. We will now show
that by using a slightly different approach the critical epsilon
can be almost doubled. To this end, we exploit the fact that only
$2M^{N-1}$ measurement settings out of all possible $M^N$ appear in
$I_{\mathsf{A}}^{N,M,d}$. However, to generate them a \textit{common} source has to be used. Assuming then that this is the case, $R=\log_2(2M^{N-1})=1+(N-1)r$ (instead of $Nr$) uses of the SV source are enough to generate all measurement settings in $I_{\mathsf{A}}^{N,M,d}$.
%
%
Thus, $Q_r(\boldsymbol{x})\leq
[(1+2\varepsilon)/(1-2\varepsilon)]^{1+(N-1)r}$, which together
with (\ref{viol}) imply that the right-hand side of (\ref{VIB})
vanishes for $M\to\infty$ iff $\varepsilon<\varepsilon_N''=(2^{1/(N-1)}-1)/[2(2^{1/(N-1)}+1)]$,
and in particular $\varepsilon''_2=1/6>\varepsilon_2'$.

\textit{Conclusions.}
We have presented a novel class of monogamy relations, obeyed by
any nonsignalling physical theory. They tightly relate the amount
of nonlocality, as quantified by the violation of Bell
inequalities \cite{BKP,rodrigo}, that $N$ parties have generated
in an experiment to the classical correlations an external party
can share with outcomes of any measurement performed by the
parties. Such trade--offs find natural applications in
device-independent protocols and here we have discussed how they
apply in quantum key distribution (cf. also Ref. \cite{2prot}) and
generation and amplification of randomness. We have finally showed
that bipartite quantum correlations allow one to amplify
$\varepsilon$--free $d$its for any $\varepsilon<1/6$.

Our results provoke further questions. First, it is natural to ask
if analogous monogamies hold for quantum correlations, and, in
fact, such elemental monogamies can be derived in the simplest
(3,2,2) scenario (see \cite{supplement}). From a more fundamental
perspective, it is of interest to understand what is the (minimal)
set of of monogamy relations generating
the same set of multipartite correlations as the no-signalling
principle.


\textit{Acknowledgments.} Discussions with Gonzalo De La Torre are
gratefully acknowledged. This work is supported by NCN grant
2013/08/M/ST2/00626, FNP TEAM, EU project SIQS, ERC grants QITBOX, QOLAPS and
QUAGATUA, the
Spanish project Chist-Era DIQIP. This publication was made possible
through the support of a grant from the
John Templeton Foundation. R. A. also
acknowledges the Spanish MINECO for the support through the Juan
de la Cierva program.

\section{appendices}

Here we present detailed proofs of Theorems 1, 2, and 3 of the
main text. Also, in the simplest $(3,2,2)$ scenario we provide
elemental monogamies for quantum correlations.

\section{Appendix A: Monogamy relations}

\subsection{Monogamy relations for nonsignalling correlations}
\label{detaliczne}

Let us start with a simple fact. Recall for this purpose that
$\langle \Omega\rangle$ is the standard mean value of a random
variable $\Omega$, that is, $\langle \Omega\rangle=\sum_{i=1}^{d-1}iP(\Omega=i)$
and
$[\Omega]$ stands for $\Omega$ modulo $d$.

\begin{faktt}\label{fact1}
It holds that for any random variable $\Omega$,
\begin{eqnarray}
\hspace{-1cm}&(a)&\hspace{1cm}\sr{\Omega}+\sr{-\Omega-1}=d-1,\\
\hspace{-1cm}&(b)&\hspace{1cm}\sr{\Omega}+\sr{-\Omega}=d[1-p([\Omega]=0)].
\end{eqnarray}
\end{faktt}
\begin{proof}Both equations follow from the very definition of
$\sr{\cdot}$. To prove (a) we notice that $[-\Omega-1]+[\Omega]=d-1$,
and hence
\begin{eqnarray}
\sr{-\Omega-1}&\nms=\nms&\sum_{i=1}^{d-1}i p([\Omega]=d-i-1)\nonumber\\
&\nms=\nms&\sum_{i=0}^{d-2}(d-i-1)p([\Omega]=i)\nonumber\\
&\nms=\nms&(d-1)\sum_{i=0}^{d-2}p([\Omega]=i)-\sum_{i=0}^{d-2}iP([\Omega]
=i)\nonumber\\
&\nms=\nms&(d-1)\sum_{i=0}^{d-1}p([\Omega]=i)-\sr{\Omega}\nonumber\\
&\nms=\nms&(d-1)-\sr{\Omega},
\end{eqnarray}
where the second equality is a consequence of changing of the
summation index, the fourth one stems from the definition of
$\sr{\Omega}$ and rearranging terms, and the last equality follows from
normalization.

To prove (b), we write
\begin{eqnarray}
\sr{\Omega}+\sr{-\Omega}&=&\sum_{i=1}^{d-1}i[p([\Omega]=i)+p([-\Omega]=i)]
\nonumber\\
&=&\sum_{i=1}^{d-1}i[p([\Omega]=i)+p([\Omega]=d-i)]\nonumber\\
&=&\sum_{i=1}^{d-1}i p([\Omega]=i)+\sum_{i=1}^{d-1}(d-i)p([\Omega]=i)\nonumber\\
&=&d\sum_{i=1}^{d-1}p([\Omega]=i)\nonumber\\
&=&d[1-p([\Omega]=0)],
\end{eqnarray}
where the second equality is a consequence of the fact that
$[\Omega]+[-\Omega]=d$,
%
while the third equality follows from shifting
of the summation index in the second sum.
\end{proof}
Let us now move to the proofs of the monogamy relations. In the
tripartite case we make use of the Barrett, Kent, and Pironio
(BKP) \cite{BKP} inequality
\begin{equation}
\label{BKP-2d}
I^{2,M,d}_{AB}=\sum_{\alpha=1}^{M}\left(\sr{A_{\alpha}-B_{\alpha}}
+\sr{B_{\alpha}-A_{\alpha+1}}\right)\geq d-1,
\end{equation}
where the convention that $X_{M+1}=[X_1+1]$ is assumed.

\setcounter{thmm}{0}

\begin{thmm}\label{thm:BKP2d:app}
For any three-partite nonsignalling correlations $\{p(a,b,c|x,y,z)\}$
with $M$ measurements and $d$ outcomes per site and any pair $\{i,j\}$
$(i,j=1,\ldots,M)$, the following inequality
\begin{equation}\label{thmm1:0}
I^{2,M,d}_{AB}+\sr{X_i-C_j}+\sr{C_j-X_i}\geq d-1
\end{equation}
is satisfied with $X$ denoting either $A$ or $B$.
\end{thmm}
\begin{proof}Let us start with the case of $X=A$
and then notice that for a random variable $\Omega$ it holds that
$\sr{\Omega}+\sr{-\Omega-1}=d-1$ (see Fact \ref{fact1}).
Consequently,
\begin{equation}\label{zero_expr}
\hspace{-0.3cm}\sum_{\substack{\beta=1\\ \beta\neq
i}}^{M}(\sr{C_j-A_{\beta}-1} +\sr{A_{\beta}-C_j})-(M-1)(d-1)
\end{equation}
is equal to zero. The fact that for any $\beta$ and $j$ it holds
that $\sr{C_j-A_{\beta}-1}
+\sr{A_{\beta}-C_j}=d-1=\sr{A_{\beta}-C_j-1} +\sr{C_j-A_{\beta}}$
allows us to rewrite (\ref{zero_expr}) in the following way
\begin{eqnarray}\label{thmm1:2}
&&\sum_{\beta=1}^{i-1}(\sr{C_j-A_{\beta}-1}
+\sr{A_{\beta+1}-C_j})\nonumber\\
&&+\hspace{-0.2cm}\sum_{\beta=i+1}^{M}(\sr{A_{\beta}-C_j-1}+\sr{C_j-A_{\beta}}
)-(M-1)(d-1).\nonumber\\
\end{eqnarray}
Then, by adding $\sr{A_i-C_j}+\sr{C_j-A_i}$ to both sides of the
above and rearranging some terms in the resulting expression, one obtains %
\begin{eqnarray}\label{thmm1:3}
&&\hspace{-0.5cm}\sr{A_i-C_j}+\sr{C_j-A_i}\nonumber\\
&&\hspace{-0.3cm}=\sum_{\beta=1}^{i-1}(\sr{C_j-A_{\beta}-1}+\sr{A_{\beta+1}-C_j}
)\nonumber\\
&&\hspace{-0.1cm}+\sum_{\beta=i}^{M-1}(\sr{A_{\beta+1}-C_j-1}+\sr{C_j-A_{\beta}}
)\nonumber\\
&&\hspace{-0.1cm}+\sr{A_1-C_j}+\sr{C_j-A_M}-(M-1)(d-1).
\end{eqnarray}
In an analogous way, we may decompose $I_{AB}^{2,M,d}$:
\begin{eqnarray}\label{thmm1:4}
I_{AB}^{2,M,d}&=&\sum_{\alpha=1}^{i-1}(\sr{A_{\alpha}-B_{\alpha}}+\sr{B_{\alpha}
-A_{\alpha+1}})\nonumber\\
&&+\sum_{\alpha=i}^{M-1}(\sr{A_{\alpha}-B_{\alpha}}+\sr{B_{\alpha}-A_{\alpha+1}}
)\nonumber\\
&&+\sr{A_M-B_M}+\sr{B_M-A_1-1}.
\end{eqnarray}
In the last step of these manipulations, we add line by line Eqs.
(\ref{thmm1:3}) and (\ref{thmm1:4}) in order to finally obtain
\begin{widetext}
\begin{eqnarray}\label{Saura}
I_{AB}^{2,M,d}+\sr{A_i-C_j}+\sr{C_j-A_i}&=&\sum_{\alpha=1}^{i-1}
(\sr{C_j-A_{\alpha}-1}+\sr{A_{\alpha}-B_{\alpha}}+\sr{B_{\alpha}-A_{\alpha+1}}
+\sr{A_{\alpha+1}-C_j})\nonumber\\
&&+\sum_{\alpha=i}^{M-1}(\sr{C_j-A_{\alpha}}+\sr{A_{\alpha}-B_{\alpha}}+\sr{B_{
\alpha}-A_{\alpha+1}}+\sr{A_{\alpha+1}-C_j-1})\nonumber\\
&&+\sr{C_j-A_M}+\sr{A_M-B_M}+\sr{B_M-A_1-1}+\sr{A_1-C_j}\nonumber\\
&&-(M-1)(d-1).
\end{eqnarray}
\end{widetext}
What we have arrived at is basically the sum of $M$ Bell
expressions $I^{2,2,d}$ but `distributed' among three parties in
such a way that Bob and Charlie measure only a single observable.
It was shown in \cite{ns-mon} that the minimal value  such an
expression can achieve over nonsignalling correlations is
precisely its classical bound $d-1$. As a
result, $I^{2,M,d}_{AB}+\sr{A_i-C_j}+\sr{C_j-A_i}\geq
M(d-1)-(M-1)(d-1)=d-1$, finishing the proof for the case
$X=A$.

If $X=B$ in Ineq. (\ref{thmm1:0}), then it suffices to rewrite the
Bell expression from (\ref{BKP-2d}) as
\begin{equation}\label{Freddie}
I^{2,M,d}_{AB}=\sum_{\alpha=1}^{M}(\sr{B_{\alpha}-A_{\alpha+1}}
+\sr{A_{\alpha+1}-B_{\alpha+1}}),
\end{equation}
add to it the zero expression (\ref{zero_expr}) with $A$ replaced by $B$,
and repeat the above manipulations. This completes the proof.
\end{proof}
Now let us move to the general $(N,M,d)$ scenario. The inequality
of interest is now the one from Ref. \cite{rodrigo}, namely:
\begin{eqnarray}\label{BKPNMd}
I^{N,M,d}_{\mathsf{A}}=\frac{1}{M}\sum_{\alpha_{N-1}=1}^{M}I_{A^{(1)}\ldots
A^{(N-1)}}^{N-1,M,d}(\alpha_{N-1})\circ A^{(N)}_{\alpha_{N-1}}\nonumber\\
&&\hspace{-1.8cm}\ge d-1.
\end{eqnarray}
where $\mathsf{A}=A^{(1)}\ldots
A^{(N)}$. The notation $\circ A^{(i)}_\gamma$ means insertion of
$A^{(i)}_\gamma$ to the average $\langle\cdot\rangle$ with the
opposite sign to the one of $A^{(i-1)}_{\delta}$ with any
$\gamma,\delta$, while $I_{A^{(1)}\ldots
A^{(N-1)}}^{N-1,M,d}(\alpha_{N-1})$ is the same Bell expression as
in (\ref{BKPNMd}), but for $N-1$ parties, and with observables of
the last party relabeled as $\alpha_{N-2}\to
\alpha_{N-2}+\alpha_{N-1}-1$ with $\alpha_N=1,\ldots,M$.
\begin{thmm}\label{thm:BKPNMd:app}
For any $(N+1)$-partite nonsignalling correlations
$\{p(\boldsymbol{a}|\boldsymbol{x})\}$ with
$M$ $d$-outcome measurements per site,
the following inequality
\begin{equation}\label{thm:BKPNMd:app1}
I^{N,M,d}_{\mathsf{A}}+\sr{A_{x_k}^{(k)}-A^{(N+1)}_{x_{N+1}}}+\sr{A^{(N+1)}_{x_{
N+1}}-A_{x_k}^{(k)}}\geq
d-1
\end{equation}
is satisfied for any $x_k,x_{N+1}=1,\ldots,M$ and  $k=1,\ldots,N$.
\end{thmm}

\begin{proof}The recursive formula in Ineq. (\ref{BKPNMd}),
which for convenience we restate here
\begin{equation}\label{recursive:app}
I^{N,M,d}_{\mathsf{A}}=\frac{1}{M}\sum_{\alpha_{N-1}=1}^{M}I_{A^{(1)}\ldots
A^{(N-1)}}^{N-1,M,d}(\alpha_{N-1})\circ A^{(N)}_{\alpha_{N-1}},
\end{equation}
allows us to demonstrate the theorem inductively. The case of
$N=2$ has already been proved as Theorem \ref{thm:BKP2d:app}, so
we consider $N=3$. Exploiting Eq. (\ref{recursive:app}), one can
express $I^{3,M,d}_{A^{(1)}A^{(2)}A^{(3)}}$ as
\begin{equation}\label{thm:BKPNMd:app2}
I^{3,M,d}_{A^{(1)}A^{(2)}A^{(3)}}=\frac{1}{M}\sum_{\alpha_2=1}^{M}I^{2,M,d}_{A^{
(1)}A^{(2)}}(\alpha_2)\circ
A^{(3)}_{\alpha_2}.
\end{equation}
It is clear that for every $\alpha_2=1,\ldots,M$
\begin{eqnarray}\label{Almodovar}
\hspace{-0.2cm}I^{2,M,d}_{A^{(1)}A^{(2)}}(\alpha_2)&=&
\sum_{\alpha_1=1}^{M}(\sr{A^{(1)}_{\alpha_1}-A^{(2)}_{\alpha_1+\alpha_2-1}}\\
&&\hspace{1cm}+\sr{A^{(2)}_{\alpha_1+\alpha_2-1}-A^{(1)}_{\alpha_1+1}})\geq
d-1\nonumber
\end{eqnarray}
is a Bell inequality equivalent to (\ref{BKP-2d}), in which the
observables of the second party $A^{(2)}$ have been relabelled
according to $\alpha_1\to \alpha_1+\alpha_2-1$. It must then
fulfil the monogamy relations (\ref{thmm1:0}) (with $N=2$)
independently of the value of $\alpha_2$. In order to see it in a
more explicit way, let us consider the case $k=1$, and in Eq.
(\ref{Saura}) just rename $A\to A^{(1)}$, $B\to A^{(2)}$, and
$C\to A^{(3)}$, and also $\alpha\to \alpha_1$ for the first party,
while $\alpha\to \alpha_1+\alpha_2-1$ for the second one. Then,
for those observables $A^{(2)}_{\alpha_1+\alpha_2-1}$ for which $
\alpha_1+\alpha_2-1>M$ we use the rule $X_{i\times
M+\gamma}=[X_{\gamma}+i]$ to get $[A_{\gamma}^{(2)}+i]$ with some
$\gamma$ and $i$, and later replace the latter by another variable
$\widetilde{A}^{(2)}_{\gamma}$ (this is just $A_{\gamma}^{(2)}$
with outcomes shifted by a constant). With the aid of formula
(\ref{Freddie}) the same reasoning can be repeated for $k=2$.

Now, we prove that each term in Eq. (\ref{thm:BKPNMd:app2})
fulfills (\ref{thm:BKPNMd:app1}) for $N=3$, that is that the inequalities
\begin{equation}\label{Queen}
I^{2,M,d}_{A^{(1)}A^{(2)}}(\alpha_2)\circ A^{(3)}_{\alpha_2}
+\sr{A_{x_k}^{(k)}-A^{(4)}_{x_4}}+\sr{A^{(4)}_{x_4}-A_{x_k}^{(k)}}\geq
d-1
\end{equation}
hold for any $\alpha_2=1,\ldots,M$, any pair $x_k,x_4=1,\ldots,M$,
and any $k=1,2,3$.

First assume $k=1$. Let us write explicitly
$I^{2,M,d}_{A^{(1)}A^{(2)}}(\alpha_2)\circ A^{(3)}_{\alpha_2}$ as
\begin{eqnarray}\label{rownanie}
&&\hspace{-0.1cm}I^{2,M,d}_{A^{(1)}A^{(2)}}(\alpha_2)\circ
A^{(3)}_{\alpha_2}=\sum_{\alpha_1=1}^M
(\sr{A^{(1)}_{\alpha_1}-A^{(2)}_{\alpha_1+\alpha_2-1}+A^{(3)}_{\alpha_2}}
\nonumber\\
&&\hspace{3.7cm}+\sr{A^{(2)}_{\alpha_1+\alpha_2-1}-A^{(1)}_{\alpha_1+1}-A^{(3)}_
{\alpha_2}}).\nonumber\\
\end{eqnarray}
For any fixed $\alpha_2$, the last party measures solely a single
observable, and therefore we treat
$A^{(2)}_{\alpha_1+\alpha_2-1}-A_{\alpha_2}^{(3)}$ as a single
variable, or, in other words, for any $\alpha_2=1,\ldots,M$,
$A^{(2)}_{\alpha_1+\alpha_2-1}-A_{\alpha_2}^{(3)}$ is a
$d$-outcome observable [recall that in Eq. (\ref{rownanie}) all
variables are modulo $d$]. Effectively, (\ref{Queen}) is a
three-partite inequality of the form (\ref{thm:BKPNMd:app1}) (with
$N=2$) that has just been proven.

In the $k=2$ case we insert the third party into the alternative
expression (\ref{Freddie}) and further apply the same reasoning as
above.

In order to show (\ref{thm:BKPNMd:app1}) for $k=3$, we use the fact that
the Bell inequality (\ref{BKPNMd}) for $N=3$ is invariant under
the exchange of the first and the third party \cite{rodrigo},
meaning that we can, analogously to Eq. (\ref{thm:BKPNMd:app2}), write it down
as
\begin{equation}\label{thm:BKPNMd:app3}
I^{3,M,d}_{A^{(1)}A^{(2)}A^{(3)}}=\frac{1}{M}\sum_{\alpha_2=1}^{M}I^{2,M,d}_{A^{
(3)}A^{(2)}}(\alpha_2)\circ
A^{(1)}_{\alpha_2}.
\end{equation}
Now, it is enough to repeat the above reasoning to complete the proof
of the monogamy relations (\ref{thm:BKPNMd:app1}) for $N=3$.

Having it proven for $N=3$, let us now assume that the theorem is
true for $N$ parties (any $N$-partite nonsignalling probability
distribution). In order to complete the proof we again refer to
the recursive formula (\ref{recursive:app}). By grouping together
the last two parties, each term in the sum in Eq.
(\ref{recursive:app}) is effectively an $(N-1)$--partite Bell
expression for which we have just assumed (\ref{thm:BKPNMd:app1})
to hold for any $x_k,x_N$ and $k=1,\ldots,N$. Performing the summation
over $\alpha_{N-1}$ and dividing further by $M^{N-2}$ we obtain
(\ref{thm:BKPNMd:app1}) for any $i,j$ and $k=1,\ldots,N-1$. The case $k=N$ can
be reached by using the fact that $I^{N,M,d}$ is invariant under exchange of the
last and the $(N-2)$th party \cite{rodrigo}.
\end{proof}


\subsection{Elemental monogamies for quantum correlations}

Let us now discuss the case of quantum correlations in which case
similar monogamy relations are also expected to hold. Their
derivation, however, is much more cumbersome and we only consider
the simplest $(3,2,2)$ scenario and derive quantum analogs of the
nonsignalling monogamies (\ref{thmm1:0}). To this end, we use a
one-parameter modification of the CHSH Bell inequality \cite{CHSH}
%
%
with the latter being a particular case of (\ref{BKP-2d}) with
$M=d=2$. Here, for convenience, we write it down in its ``standard" form:
\begin{equation}\label{CHSHalfa}
\widetilde{I}^{\alpha}_{AB}:=\alpha(\langle A_1B_1\rangle+\langle
A_1B_2\rangle)+\langle A_2B_1\rangle-\langle A_2B_2\rangle\leq
2\alpha
\end{equation}
with $\alpha\geq 1$. Here, $A_i$ and $B_i$ are local
quantum observables with eigenvalues $\pm 1$ and $\langle
XY\rangle=\Tr[\rho(X\otimes Y)]$ for some state $\rho$ and local
observables $X,Y$. Actually, one proves the following more general theorem,
generalizing the result of Ref. \cite{TV} for the Bell inequality
(\ref{CHSHalfa}).
%
%

\begin{thmm}\label{thm:Qapp}
Any three-partite quantum correlations with two dichotomic
measurements per site must satisfy the following inequalities
\begin{eqnarray}\label{in1}
&&\alpha^2\max\{(\widetilde{I}^{\alpha}_{AB})^2,(\widetilde{I}^{\alpha}_{AC})^2
\}
+\min\{(\widetilde{I}^{\alpha}_{AB})^2,(\widetilde{I}^{\alpha}_{AC})^2\}
\nonumber\\
&&\hspace{4cm}\leq4\alpha^2(1+\alpha^2)
\end{eqnarray}
and
\begin{equation}\label{in2}
(\widetilde{I}_{AB}^{\alpha})^2+4\langle A_iC_j\rangle^2\leq
4(1+\alpha^2)
\end{equation}
for any $\alpha\geq 1$ and $i,j=1,2$.
\end{thmm}
\begin{proof}The proof is nothing more but a slight modification of
the considerations of Ref. \cite{TV} (see also Ref.
\cite{Horodeccy}). Nevertheless, we attach it here for
completeness.

We start by noting that the monogamy regions, that is, the
two-dimensional sets of allowed (realizable) within quantum theory
pairs
$\{\widetilde{I}_{AB}^{\alpha},\widetilde{I}_{AC}^{\alpha}\}$ for
Ineq. (\ref{in1}) and $\{\widetilde{I}_{AB}^{\alpha},\langle
A_iC_j\rangle\}$ with fixed $i$ and $j$ for Ineq. (\ref{in2}),
must be convex. Therefore, as it is shown in Ref. \cite{TV} (see
also Ref. \cite{Lluis}), every point of their boundaries can be
realized with a real three-qubit pure state and real local
one-qubit measurements. Recall that the latter assumes the form
\begin{equation}\label{theform}
X=\boldsymbol{x}\cdot\boldsymbol{\sigma}
\end{equation}
with $\boldsymbol{x}\in\mathbbm{R}^2$ being a unit vector and
$\boldsymbol{\sigma}=[\sigma_x,\sigma_z]$ denoting a vector
consisting of the standard Pauli matrices $\sigma_x$ and
$\sigma_z$.

Then, it follows from a series of papers \cite{Horodeccy,TV,AMP}
that for a given two-qubit state $\rho_{AB}$, the maximal value of
$\widetilde{I}_{AB^{\alpha}}$ over local, real, and traceless
observables [i.e., those of the form (\ref{theform})] measured by
Alice $A_i$ and Bob $B_i$, amounts to
\begin{equation}\label{ElPrat}
\max_{A_i,B_j}(\widetilde{I}^{\alpha}_{AB})=2\sqrt{\alpha^2
\lambda_1+\lambda_2}.
\end{equation}
Here, $\lambda_i$ $(i=1,2)$ denote the eigenvalues of
$T_{AB}T_{AB}^{T}$ put in a decreasing order, i.e.,
$\lambda_1\geq\lambda_2$, and $T_{AB}$ is the following ŽreducedŽ
correlation matrix
\begin{equation}\label{matrix}
T_{AB}=\left(
\begin{array}{cc}
\langle \sigma_x\otimes\sigma_x\rangle_{AB} & \langle
\sigma_x\otimes\sigma_z\rangle_{AB}\\
\langle \sigma_z\otimes\sigma_x\rangle_{AB} & \langle
\sigma_z\otimes\sigma_z\rangle_{AB}
\end{array}
\right).
\end{equation}
We added the subscript $AB$ in (\ref{matrix}) to indicate that the
mean values are taken in the state $\rho_{AB}$. In particular, one
can similarly compute the maximal value of a single average
$\langle AB\rangle$ in the state $\rho_{AB}$ over local
observables $A$ and $B$ of the form (\ref{theform}) to be
\begin{equation}
\max_{A,B}\langle AB\rangle=\lambda_1.
\end{equation}

Equipped with these facts, we can now turn to the proof of the
inequalities (\ref{in1}) and (\ref{in2}). We start from the first
one and note that it suffices to demonstrate it in the case of
$\widetilde{I}_{AB}^{\alpha}\geq \widetilde{I}_{AC}^{\alpha}$, in
which it becomes
\begin{equation}\label{tintodeverano}
\alpha^2(\widetilde{I}_{AB}^{\alpha})^2+(\widetilde{I}_{AC}^{\alpha})^2\leq
4\alpha^2.
\end{equation}
The opposite case will follow immediately by exchanging
$B\leftrightarrow C$.

Let then $\ket{\psi_{ABC}}$ be a pure real three-qubit state. By
$\rho_{AB}$ and $\rho_{AC}$ we denote its subsystems arising by
tracing out the third and the second party, respectively, and by
$T_{AB}$ and $T_{AC}$ the corresponding correlation matrices [cf.
Eq. (\ref{matrix})]. Finally, let $\lambda_i$ and
$\widetilde{\lambda}_i$ $(i=1,2)$ be eigenvalues of
$T_{AB}T_{AB}^T$ and $T_{AC}T_{AC}^T$, respectively, where we keep
the convention that $\lambda_1\geq \lambda_2$ and
$\widetilde{\lambda}_1\geq \widetilde{\lambda}_2$. It was pointed
out in Ref. \cite{TV} that the latter matrices are diagonal in the
same basis, which allows one to simultaneously maximize both
$\widetilde{I}_{AB}^{\alpha}$ and $\widetilde{I}_{AC}^{\alpha}$
with the same observables on Alice site. This, together with Eq.
(\ref{ElPrat}), means that
\begin{eqnarray}\label{girona}
\hspace{-0.1cm}\max_{\substack{A_i,B_j,\\C_k}}\!\!\left[\alpha^2(\widetilde{I}_{
AB}^{\alpha})^2+(\widetilde{I}_{AC}^{\alpha})^2\right]
&\negmedspace=\negmedspace&
4[\alpha^2(\alpha^2\lambda_1+\lambda_2)+\alpha^2\widetilde{\lambda}_1+\widetilde
{\lambda}_2]\nonumber\\
&\negmedspace=\negmedspace&4[\alpha^4\lambda_1+\alpha^2(\lambda_2+\widetilde{
\lambda}_1)+\widetilde{\lambda}_2].\nonumber\\
\end{eqnarray}
In order to complete the proof, we make use of the
Toner-Verstraete monogamy relation for the CHSH Bell inequality
\cite{TV}, which we state here in terms of $\lambda_i$ and
$\widetilde{\lambda}_i$ as
\begin{equation}\label{CHSHmonL}
\lambda_2+\widetilde{\lambda}_1\leq
2-\lambda_1-\widetilde{\lambda}_2.
\end{equation}
When applied to (\ref{girona}), it leads us to
\begin{eqnarray}\label{girona2}
\hspace{-0.1cm}\max_{\substack{A_i,B_j,\\C_k}}\!\left[\alpha^2(\widetilde{I}_{AB
}^{\alpha})^2+(\widetilde{I}_{AC}^{\alpha})^2\right]
&\negmedspace\leq\negmedspace&4[(\alpha^2-1)(\alpha^2\lambda_1-\widetilde{
\lambda}_2)+2\alpha^2]\nonumber\\
&\negmedspace=\negmedspace&4[\alpha^2(\alpha^2-1)+2\alpha^2]\nonumber\\
&\negmedspace=\negmedspace&4\alpha^2(1+\alpha^2),
\end{eqnarray}
where the second line follows form the facts that $\lambda_1\leq
1$, $\widetilde{\lambda}_2\geq 0$, and $\alpha\geq 1$.

To prove Ineq. (\ref{in2}), we follow the above reasoning to
obtain
\begin{eqnarray}
\max_{A_i,B_j,C_l}\left[(\widetilde{I}^{\alpha}_{AB})^2+4\langle
A_kC_l\rangle^2\right]&=&4(\alpha^2\lambda_1+\lambda_2)+4\widetilde{\lambda}
_1\nonumber\\
&=&4\alpha^2\lambda_1+4(\lambda_2+\widetilde{\lambda}_1)\nonumber\\
\end{eqnarray}
for $k=1,2$. Subsequent application of (\ref{CHSHmonL}) to the
term in parentheses in the second line of the above directly gives
Ineq. (\ref{in2}), completing the proof.
\end{proof}

For $i=1$ and $j=1,2$, the relations (\ref{in2}) are tight as any pair of
values of $\widetilde{I}_{AB}^{\alpha}$ and $\langle
A_1C_j\rangle$ saturating them can be realized with the state
$(\beta_{+}\ket{01}+\beta_{-}\ket{10})\ket{0}$, where
$\beta_{\pm}=\tfrac{1}{2}(1\pm \sqrt{2}\sin\theta)^{1/2}$ and
$\theta\in[0,\pi/4]$. It is, however, no longer true for $i=2$. In
this case we numerically found tight monogamy relations for
particular values of $\alpha$ (see Fig. \ref{fig:tight}).
\begin{figure}[h!]
(a)\includegraphics[width=0.35\textwidth]{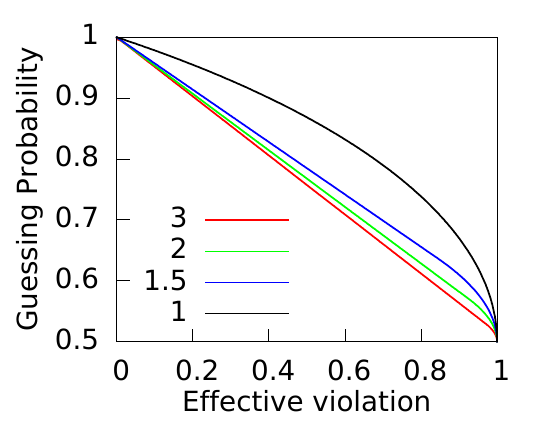}
(b)\includegraphics[width=0.35\textwidth]{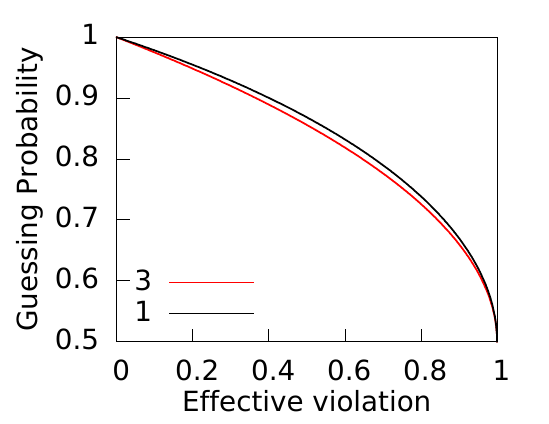}
\caption{(a) Guessing probability (and simultaneously the tight
analogs of monogamies in Theorem \ref{thm:Qapp}) for $i=2$
as a function of
$(\widetilde{I}_{AB}^{\alpha}-2\alpha)/2(\sqrt{1+\alpha^2}-\alpha)$
for various values of $\alpha$. All curves were found using two methods. First,
we maximized the guessing probability for a given
value of $\widetilde{I}_{AB}^{\alpha}$ over two-ququart states and
one-ququart dichotomic measurements. Then, we used the hierarchy
of Ref. \cite{NPA} and with its third level we arrived at curves
that coincide with those obtained with the first method with
precision $10^{-7}$. For comparison (b) presents the
corresponding nontight monogamies proven in theorem \ref{thm:Qapp} ($i=2$) for
$\alpha=1,3$ (the curves for $\alpha=1.5,2$ fall in between these
two). The black curve is the same on both plots.}\label{fig:tight}
\end{figure}


Let us finally notice that the quantum elemental monogamies
impose the following upper bounds on the guessing probability
\begin{equation}
\max_j p(X_i=j)\leq
\frac{1}{2}\{1+[1+\alpha^2-(\widetilde{I}_{AB}^{\alpha}/2)^2]^{1/2}\}
\end{equation}
with $X=A,B$, $i=1,2$, and $\alpha\geq 1$. This bound was already
derived in Ref. \cite{AMP}, and, as already said, it is tight only
for $i=1$. In the case $i=2$, we determined the tight bounds
numerically for few $\alpha$s and they are presented in Fig. \ref{fig:tight}.

%
%

\section{Appendix B: Randomness amplification}
Let us begin with recalling the description of the
Santha--Vazirani (SV) source (or the $\varepsilon$--source). Its
working is defined as follows: it produces a sequence $y_1,
y_2,\cdots, y_n$ of bits according to
\begin{equation}\label{sv}
\tfrac{1}{2}-\varepsilon \le p(y_k| w)\le \tfrac{1}{2}+\varepsilon
\quad (k=1,\ldots,n),
\end{equation}
where $w$ denotes any space-time variable that could be the cause
of $y_k$. In particular, $y_k$ can depend on $y_1,\ldots,y_{k-1}$.

Let now $W$ be any random variable used by an adversary to control
the $\varepsilon$-sources and the physical systems held by the
parties. The random variable can be thought of a device, held by a
villain $E$, with a knob that when set to a particular value $w$
of $W$ makes (i) the SV sources produce bits with certain
probabilities obeying (\ref{sv}) and (ii) the devices held by the
parties generate a concrete nonsignalling probability distribution
represented by
$\{p(\boldsymbol{a}|\boldsymbol{x},w)\}_{\boldsymbol{a},\boldsymbol{x}}$. Let us
then by $\{p(a_k,w|\boldsymbol{x})\}_{a_k,w}$ denote correlations between
outcomes obtained by party $k$ and the random variable $W$ for a
particular choice of measurement settings $\boldsymbol{x}$.
Also, let $\{\widetilde{p}(a)\}$ be the one-party uniform
probability distribution, i.e., $\widetilde{p}(a)=1/d$ for any
$a$. Introducing then the variational distance
\begin{equation}
D(\{p(x)\},\{q(x)\})=\frac{1}{2}\sum_x|p(x)-q(x)|
\end{equation}
between two probability distributions $\{p(x)\}$
and $\{q(x)\}$, we can prove the following.

\begin{thmm} Let for any $w$,
$\{p(\boldsymbol{a}|\boldsymbol{x},w)\}_{\boldsymbol{a},\boldsymbol{x}}$ be an
$N$-partite nonsignalling probability distribution. Then for any $k=1,\ldots,N$
and any choice of
measurement settings $\boldsymbol{x}$:
\begin{eqnarray}\label{GaiaNova2}
&&\hspace{-1.cm}D(\{p(a_k,w|\boldsymbol{x})\}_{a_k,w},\{\widetilde{p}
(a_k)p(w|\boldsymbol { x } )\}_{a_k,w})\nonumber\\
&&\hspace{0.5cm}=\frac{1}{2}\sum_{a_k,w}\left|p(a_k,w|\boldsymbol{x}
)-\widetilde{p}(a_k)p(w|\boldsymbol{x})\right|\nonumber\\
&&\hspace{0.5cm}\leq
\frac{(d-1)^2+1}{2d}\,Q_M(\boldsymbol{x})I^{N,M,d}_{\mathsf{A}},
\end{eqnarray}
where $I^{N,M,d}_{\mathsf{A}}$ is taken in the
probability distribution observed by the parties
$\{p(\boldsymbol{a}|\boldsymbol{x})\}$. Then
\begin{equation}
Q_M(\boldsymbol{x})=\max_{w}\left[\frac{p(w|\boldsymbol{x})}{p_{\min}(w)}
\right ],
\end{equation}
where $p_{\min}(w)=\min_{\boldsymbol{x}}\{p(w|\boldsymbol{x})\}$
with the minimum taken over all measurement settings
$\boldsymbol{x}$ appearing in the Bell inequality (\ref{BKPNMd}).

\end{thmm}

\begin{proof} For simplicity, but without any loss of generality, we prove
this theorem for the bipartite case. The generalization to the
multipartite case is straightforward.

As before, we denote the parties by $A$ and $B$,
while the adversary by $E$. Then, the corresponding inputs and outputs are
denoted by $x$, $y$, $z$, and $a$, $b$, and $e$, respectively.

Let us start by noting that for any probability distribution
$\{p(a,b|x,y,w)\}_{a,b,x,y}$, the maximal probability of local outcomes
obtained by any of the parties, say for simplicity Alice, must
obey the inequalities on the guessing probability [see Ineq. (7) in
the main text]. That is
\begin{equation}\label{Barbastro2}
\max_{a}p(a|x,w)\leq \frac{1}{d}\left(1+I^{2,M,d}_w
\right)
\end{equation}
for any $x=1,\ldots,M$, where by $I_w^{2,M,d}$ we have denoted the
value of the Bell expression (\ref{BKP-2d}) computed for the
probability distribution $\{p(a,b|x,y,w)\}_{a,b,x,y}$. Clearly, this bound
holds also for any $p(a|x,w)$ which together with the
normalization
\begin{equation}
p(a|x,w)=1-\sum_{\alpha\neq a}p(\alpha|x,w),
\end{equation}
means that $p(a|x,w)\geq (1/d)[1-(d-1)I^{2,M,d}_w]$,
and therefore the inequality
\begin{equation}\label{Campo2}
\left|p(a|x,w)-\frac{1}{d}\right|\leq
\frac{d-1}{d}I_w^{2,M,d}
\end{equation}
holds for any $a$ and $x$. Using then the inequality (\ref{Barbastro2}) for
$\max_ap(a|x,w)$ and (\ref{Campo2}) for the rest of $p(a|x,w)$,
we obtain that for any strategy $w$ and a measurement setting $x$,
\begin{eqnarray}\label{Guimeraes2}
D(\{p(a|x,w)\}_{a},\{\widetilde{p}(a)\})
&=&\frac{1}{2}\sum_{a}\left|p(a|x,w)-\widetilde{p}(a)\right|\nonumber\\
&\leq&\frac{(d-1)^2+1}{2d}I_w^{2,M,d}.
\end{eqnarray}

The remainder of the proof goes along exactly the same lines as in Ref.
\cite{collbeck-renner}, however, for completeness, we will recall it here.

Due to the fact that the observers do not have access to the variable
$W$, one has to average Ineq. (\ref{Guimeraes2}) over the probability
distribution $\{p(w|x,y)\}_w$ for a particular choice of measurements $x$ and
$y$.
Together with the facts that $p(a|x,w)=p(a|x,y,w)$ (no-signalling) and
$p(w|x,y)p(a|x,y,w)=p(a,w|x,y)$, this allows one to write
\begin{eqnarray}\label{ine332}
&&\hspace{-1cm}D(\{p(a,w|x,y)\}_{a,w},\{\widetilde{p}(a)p(w|x,y)\}_{a,w}
)\nonumber\\
&&=\frac{1}{2}\sum_{a,w}\left|p(a,w|x,y)-\widetilde{p}(a)p(w|x,
y)\right|\nonumber\\
&&\leq \frac{(d-1)^2+1}{2d}\sum_w
p(w|x,y)I^{2,M,d}_w
\end{eqnarray}

Let us now concentrate on the right-and side of Ineq.
(\ref{ine332}). By using Eq. (\ref{BKP-2d}), we can bound it from
above in the following way
\begin{eqnarray}\label{nier2}
&&\sum_w p(w|x,y)I^{2,M,d}_w=\nonumber\\
&&\sum_{w,\alpha}p(w|x,y)(\sr{A_{\alpha}-B_{\alpha}}_w+\sr{B_{\alpha}-A_{
\alpha+1}}_w)\nonumber\\
&&=\sum_{w,\alpha}\left(p(w|\alpha,\alpha)\frac{p(w|x,y)}{p(w|\alpha,\alpha)}\sr
{A_{\alpha}-B_{\alpha}}_w\right.\nonumber\\
&&\hspace{1.5cm}+\left.p(w|\alpha+1,\alpha)\frac{p(w|x,y)}{p(w|\alpha+1,\alpha)}
\sr{B_{\alpha}-A_{\alpha+1}}_w\right)\nonumber\\
&&\leq
Q_M(x,y)\sum_{w,\alpha}\left[p(w|\alpha,\alpha)\sr{A_{\alpha}-B_{\alpha}}
_w\right.\nonumber\\
&&\hspace{2.5cm}+\left.p(w|\alpha+1,\alpha)\sr{B_{\alpha}-A_{\alpha+1}}
_w\right]\nonumber\\
&&=Q_M(x,y)\sum_{\alpha}(\sr{A_{\alpha}-B_{\alpha}}+\sr{B_{\alpha}-A_{\alpha+1}}
)\nonumber\\
&&=Q_M(x,y)I^{2,M,d}_{AB},
\end{eqnarray}
where the subscript $w$ in the expectation values $\sr{A_{\alpha}-B_{\alpha}}_w$
and $\sr{B_{\alpha}-A_{\alpha+1}}_w$
means that they are computed for the probability distribution
$\{p(a,b|x,y,w)\}_{a,b,x,y}$, and also the convention $p(M+1,M|w)\equiv
p(1,M|w)$ is used.
Then, $I^{2,M,d}_{AB}$ is computed for the probability distribution
$\{p(a,b|x,y)\}$ observed by $A$ and $B$.

By substituting Ineq. (\ref{nier2}) to Ineq. (\ref{ine332}),
one finally obtains Ineq. (\ref{GaiaNova2}), completing the proof.
\end{proof}
One then recovers the inequality of Ref. \cite{collbeck-renner}
from Ineq. (\ref{GaiaNova2}) by exploiting the fact that $[(d-1)^2+1]/d\leq d-1$
$(d\geq2)$.
%
%
Let us also notice that one can derive
Ineq. (\ref{GaiaNova2}) using a slightly different approach, which,
for completeness, we present below.

\begin{thmm} Let
$\{p(\boldsymbol{a}|\boldsymbol{x},w)\}_{\boldsymbol{a},\boldsymbol{x}}$ be a
nonsignalling
probability distribution for any $w$ and let the probabilities
$p(\boldsymbol{x})$ be all equal. Then for any $k=1,\ldots,N$
and any choice of measurement settings $\boldsymbol{x}$:
\begin{eqnarray}\label{GaiaNova}
&&\hspace{-1cm}D(\{p(a_k,w|\boldsymbol{x})\}_{a_k,w},\{\widetilde{p}
(a_k)p(w|\boldsymbol
{ x } )\}_{a_k,w} )\nonumber\\
&&\hspace{0.5cm} =\frac{1}{2}\sum_{a_k,w}\left|p(a_k,w|\boldsymbol{x}
)-\widetilde{p}(a_k)p(w|\boldsymbol{x})\right|\nonumber\\
&&\hspace{0.5cm}\leq
\frac{(d-1)^2+1}{2d}\,\widetilde{Q}_M(\boldsymbol{x})I^
{N,M,d}_{\mathsf{A}},
\end{eqnarray}
where $I^{N,M,d}_{\mathsf{A}}$ is taken in the
probability distribution observed by the parties
$\{p(\boldsymbol{a}|\boldsymbol{x})\}$ and
\begin{equation}
\widetilde{Q}_M(\boldsymbol{x})=
\max_{w}\left[\frac{p(\boldsymbol{x}|w)}{\widetilde{p}_{\min}(w)}
\right ],
\end{equation}
where
$\widetilde{p}_{\min}(w)=\min_{\boldsymbol{x}}\{p(\boldsymbol{x}|w)\}$
with the minimum taken over all measurement settings
$\boldsymbol{x}$ appearing in the Bell inequality (\ref{BKPNMd}).

\end{thmm}
\begin{proof}For simplicity but without any loss of generality, we prove
this theorem for the bipartite case. The generalization to the multipartite case
is straightforward.

As before, we denote the parties by $A$ and $B$,
while the adversary by $E$. Then, the corresponding inputs and outputs are
denoted by $x$, $y$, $z$, and $a$, $b$, and $e$, respectively.

Let us start by noting that, by analogy to the case considered in
the main text [see Ineq. (6) there], for any $w$, the probability
distribution $\{p(a,b|x,y,w)\}_{a,b,x,y}$ satisfies the following monogamy
relations
\begin{equation}\label{monog_mod}
\frac{I_{w}^{2,M,d}}{\widetilde{p}_{\min}(w)}+1\geq d p(X_i=E_j|w) \qquad
(X=A,B)
\end{equation}
for any pair $\{i,j\}$ $(i,j=1,\ldots,M)$. In the above
\begin{eqnarray}
\hspace{-1cm}I_w^{2,M,d}&=&\sum_{\alpha=1}^{M}[p(\alpha,\alpha|w)\sr{A_{\alpha}
-B_ {\alpha} }
_w\nonumber\\
&&\hspace{0.5cm}+p(\alpha+1,\alpha|w)\sr { B_ { \alpha} -A_ { \alpha+1} }_w ],
\end{eqnarray}
is a modified BKP Bell expression taking into account that the
inputs $x,y$ are generated with the biased probabilities
$p(x,y|w)$, all correlators $\sr{A_{\alpha} -B_ {\alpha} }_w$
and $\sr { B_{ \alpha } -A_ { \alpha+1} }_w$ are computed for the distribution
$\{p(a,b|x,y,w)\}_{a,b,x,y}$, and now
\begin{equation}
\widetilde{p}_{\min}(w)=\min_{\alpha=1,\ldots,M}\{p(\alpha,\alpha|w),p(\alpha+1,
\alpha|w)\},
\end{equation}
where the convention $p(M+1,M|w)\equiv p(1,M|w)$ is used.

The monogamy relations (\ref{monog_mod}) imply (see the main text
for the argument in favor of this fact) the bound on the
probability of the adversary when using the strategy $w$ to guess
the outcomes of any of the measurements performed by one of the
parties, say for concreteness Alice (but the same bound holds for
outcomes of party $B$):
\begin{equation}\label{Barbastro}
\max_a p(a|x,w)\leq
\frac{1}{d}\left(1+\frac{I^{2,M,d}_w}{\widetilde{p}_{\min}(w)}\right)\quad
(x=1,\ldots,M).
\end{equation}

Clearly, this bound holds also for any $p(a|x,w)$ which
together with the normalization
\begin{equation}
p(a|x,w)=1-\sum_{\alpha\neq a}p(\alpha|x,w),
\end{equation}
mean that $p(a|x,w)\geq (1/d)-(d-1)(I^{2,M,d}_w/d\widetilde{p}_{\min}(w))$,
and therefore the inequality
\begin{equation}\label{Campo}
\left|p(a|x,w)-\frac{1}{d}\right|\leq
\frac{d-1}{d}\frac{I_w^{2,M,d}}{\widetilde{p}_{\min}(w)}.
\end{equation}
holds for any $a$ and $x$. Using then the inequality (\ref{Barbastro}) for
$\max_ap(a|x,w)$ and (\ref{Campo}) for the rest of $p(a|x,w)$, we obtain that
for any strategy $w$,
\begin{eqnarray}\label{Guimeraes}
D(\{p(a|x,w)\}_{a},\{\widetilde{p}(a)\})&=&\frac{1}{2}\sum_{a}|p(a|x,
w)-\widetilde{p}(a)|\nonumber\\
&\leq&
\frac{(d-1)^2+1}{2d}\frac{I_w^{2,M,d}}{\widetilde{p}_{\min}(w)}.
\end{eqnarray}
Now, since the parties do not have access to $W$, one needs further to
average Ineq. (\ref{Guimeraes}) over the probability distribution
$\{p(w|x,y)\}_w$ for a particular choice of measurements $x$ and $y$.
This, together with the facts that $p(a|x,w)=p(a|x,y,w)$ (no-signalling)
and $p(w|x,y)=p(x,y|w)p(w)/p(x,y)$ implying that
$p(w|x,y)p(a|x,y,w)=p(a,w|x,y)$, allows one to write

\begin{eqnarray}\label{ine33}
&&\hspace{-1cm}D(\{p(a,w|x,y)\}_{a,w},\{\widetilde{p}(a)p(w|x,y)\}_{a,w}
)\nonumber\\
&&=\frac{1}{2}\sum_{a,w}\left|p(a,w|x,y)-\widetilde{p}(a)p(w|x,
y)\right|\nonumber\\
&&\leq \frac{(d-1)^2+1}{2d}\sum_w
\frac{p(x,y|w)}{\widetilde{p}_{\min}(w)}\frac{p(w)}{p(x,y)}I^{2,M,d}
_w\nonumber\\
&&\leq
\frac{(d-1)^2+1}{2d}\widetilde{Q}_M(x,y)\sum_w\frac{p(w)}{p(x,y)}I^{2,M,d}_w,
\end{eqnarray}
with $\widetilde{Q}_M(x,y)=\max_{w}\left[p(x,y|w)/\widetilde{p}_{\min}(w)\right]
.$ In order to obtain Ineq. (\ref{GaiaNova}) from Ineq. (\ref{ine33})
it is enough to notice that
\begin{equation}
p(a,b|x,y)=\sum_{w}p(w|x,y)p(a,b|x,y,w)
\end{equation}
which, with the aid of the assumption that all the probabilities $p(x,y)$ are
equal, further translates to
\begin{equation}
I^{2,M,d}_{AB}=\sum_{w}\frac{p(w)}{p(x,y)}I^{2,M,d}_w,
\end{equation}
where $I^{2,M,d}_{AB}$ is computed for the observed probability distribution
$\{p(a,b|x,y)\}$ and the probabilities $p(x,y)=\sum_wp(w)p(x,y|w)$ are assumed
to be equal for all $x,y$. This completes the proof.
\end{proof}
Let us finally notice that under the assumption, which we make above, that
all $p(x,y)$ are equal, it holds that
$Q_M(\boldsymbol{x})=\widetilde{Q}_M(\boldsymbol{x})$.

\end{document}